\documentclass{amsart}
\usepackage{amsfonts}
\newtheorem{theorem}{Theorem}[section]
\newtheorem{lemma}[theorem]{Lemma}
\newtheorem{proposition}[theorem]{Proposition}
\newtheorem{corollary}[theorem]{Corollary}
\theoremstyle{definition}
\newtheorem{definition}[theorem]{Definition}
\newtheorem{example}[theorem]{Example}

\theoremstyle{remark}

\numberwithin{equation}{section}

\begin{document}
\title{Quadratic Residue Codes over $\mathbb{F}_p+v\mathbb{F}_p$ and Their
Gray Images}
\author{Abidin Kaya}
\address{Abidin Kaya, Department of Mathematics, Fatih University, 34500,
Istanbul, Turkey }
\email{\texttt{akaya@fatih.edu.tr}}
\author{Bahattin Yildiz}
\address{Bahattin Yildiz, Department of Mathematics, Fatih University,
34500, Istanbul, Turkey }
\email{\texttt{byildiz@fatih.edu.tr}}
\author{Irfan Siap}
\address{Irfan Siap, Department of Mathematics, Yildiz Technical University,
34210, Istanbul, Turkey }
\email{\texttt{isiap@yildiz.edu.tr}}
\date{June, 2012}

\begin{abstract}
In this paper quadratic residue codes over the ring $\mathbb{F}_p+v\mathbb{F}%
_p$ are introduced in terms of their idempotent generators. The structure of
these codes is studied and it is observed that these codes share similar
properties with quadratic residue codes over finite fields. For the case $%
p=2 $, Euclidean and Hermitian self-dual families of codes as extended
quadratic residue codes are considered and two optimal Hermitian self-dual
codes are obtained as examples. Moreover, a substantial number of good $p$%
-ary codes are obtained as images of quadratic residue codes over $\mathbb{F}%
_p+v\mathbb{F}_p$ in the cases where $p$ is an odd prime. These results are
presented in tables.
\end{abstract}

\maketitle

\section{Introduction}

Quadratic residue codes fall into the family of BCH codes and have proven to
be a promising family of cyclic codes. They were first introduced by Andrew
Gleason and since then have generated a lot of interest. Two famous members
of this family which are Hamming and Golay codes are optimal and perfect
codes. Further, as extended quadratic residue cyclic codes, extended Hamming
and Golay codes also have been and are being applied in many applications in
data transmission such as Voyager (1979-1981), and more recently \cite%
{bae,dono,garcia,trots}. While initially quadratic residue codes were
studied within the confines of finite fields, there have been recent
developments in the form of quadratic residue codes over some special rings.

First, Pless and Qian studied quaternary quadratic residue codes (over the
ring $\mathbb{Z}_4$) and some of their properties in \cite{Pless}. In 2000,
Chiu et al. extended the ideas in \cite{Pless} to the ring $\mathbb{Z}_8$ in
\cite{chiu}. Recently, Taeri considered quadratic residue codes over the
ring $\mathbb{Z}_9$ in \cite{Taeri}.

Our aim in this paper is to introduce and study quadratic residue codes over
the ring $\mathbb{F}_{p}+v\mathbb{F}_{p}$ which is isomorphic to $\mathbb{F}%
_{p}\times \mathbb{F}_{p}$. Codes over $\mathbb{F}_{p}+v\mathbb{F}_{p}$ were
first introduced by Bachoc in \cite{Bachoc} together with a new weight. They
are shown to be connected to lattices and have since generated interest
among coding theorists. For some of the work in the literature about these
codes we refer the readers to \cite{Bachoc,betsumiya,Bonnecaze,Dougherty1}
and \cite{Harada}. Recently, Zhu et al. considered the structure of cyclic
codes over $\mathbb{F}_{2}+v\mathbb{F}_{2}$ in \cite{Zhu}.

We first give some preliminaries about the ring $\mathbb{F}_{p}+v\mathbb{F}%
_{p}$ and codes over $\mathbb{F}_{p}+v\mathbb{F}_{p}$ in Section 2. In
Section 3, quadratic residue codes over the ring $\mathbb{F}_{p}+v\mathbb{F}%
_{p}$ are defined and it is shown that they share similar properties with
quadratic residue codes over fields. In Section 4, we study the case $p=2$
and obtain Euclidean self-dual codes for $q=8r-1$ and Hermitian self-dual
codes for $q=8r+1$ as the extended quadratic residue codes over $\mathbb{F}%
_{2}+v\mathbb{F}_{2}$. Two well-known optimal Hermitian self-dual codes over
$\mathbb{F}_{2}+v\mathbb{F}_{2}$ are obtained from quadratic residue codes.

Section 5 includes the examples of QR codes for odd prime $p$. A number of
examples including best known, optimal self-dual and extremal $p$-ary codes
are obtained as Gray images of QR codes and the results are presented in
tables.

\section{Preliminaries}

The ring $\mathbb{F}_{p}+v\mathbb{F}_{p}$ is a commutative ring of order $%
p^{2}$ and characteristic $p$, subject to the restriction $v^{2}=v$. It is
easily observed that the ring $\mathbb{F}_{p}+v\mathbb{F}_{p}$ is isomorphic
to the ring $\mathbb{F}_{p}\times \mathbb{F}_{p}$. It has two maximal ideals
$\left\langle v\right\rangle $ and $\left\langle 1-v\right\rangle $. So, it
is not a local ring. The ring is Frobenius, and hence is suitable for
studying codes. A code $C$ of length $n$ over $\mathbb{F}_{p}+v\mathbb{F}%
_{p} $ is an $\left( \mathbb{F}_{p}+v\mathbb{F}_{p}\right) $-submodule of $%
\left( \mathbb{F}_{p}+v\mathbb{F}_{p}\right) ^{n}$. An element of $C$ is
called a codeword of $C$. A generator matrix of $C$ is a matrix whose rows
generate $C $. The Hamming weight of a codeword is the number of non-zero
components. Let $x=\left( x_{1},x_{2},\ldots ,x_{n}\right) $ and $y=\left(
y_{1},y_{2},\ldots ,y_{n}\right) $ be two elements of $\left( \mathbb{F}%
_{p}+v\mathbb{F}_{p}\right) ^{n}$. The Euclidean inner product is given as $%
\left\langle x,y\right\rangle _{E}=\sum x_{i}y_{i}$. The dual code $C^{\bot
} $ of $C$ with respect to the Euclidean inner product is defined as
\begin{equation*}
C^{\bot }=\left\{ x\in \left( \mathbb{F}_{p}+v\mathbb{F}_{p}\right) ^{n}\mid
\left\langle x,y\right\rangle _{E}=0\text{ for all }y\in C\right\}
\end{equation*}%
$C$ is self-dual if $C=C^{\bot }$.

In the sequel we let $R_{p,n}:=\left( \mathbb{F}_{p}+v\mathbb{F}_{p}\right) %
\left[ x\right] /\left( x^{n}-1\right) $. Where there is no confusion a
polynomial $f\left( x\right) $ is abbreviated as $f$. An idempotent is an
element $a(x)$ such that $a(x)^{2}=a(x)$. We characterize all idempotents in
$R_{p,n}$ in the following lemma:

\begin{lemma}
\label{idem}$\left\{ \left( 1-v\right) f+vh\mid f\text{ and }h\text{ are
idempotents in }\mathbb{F}_{p}\left[ x\right] /\left( x^{n}-1\right)
\right\} $ is the set of all idempotents in $R_{p,n}$.
\end{lemma}

\begin{proof}
Let $g=\left( 1-v\right) f+vh$ be an arbitrary idempotent in $R_{p,n}$ then,
\begin{eqnarray*}
g &=&g^{2}=\left( 1-v\right) ^{2}f{}^{2}+v^{2}h{}^{2}\text{ (since }v\left(
1-v\right) =0\text{) and so we get,} \\
\left( 1-v\right) f+vh &=&\left( 1-v\right) f{}^{2}+vh{}^{2}\text{,}
\end{eqnarray*}%
which implies that $f$ and $h$ are idempotents in $\mathbb{F}%
_{p}[x]/(x^{n}-1)$.

Conversely, if $f$ and $h$ are idempotents in $\mathbb{F}_{p}[x]/(x^{n}-1)$
then $\left( 1-v\right) f+vh$ is an idempotent in\ $R_{p,n}$, since
\begin{eqnarray*}
\left[ \left( 1-v\right) f+vh\right] ^{2} &=&\left( 1-v\right)
^{2}f{}^{2}+v^{2}h{}^{2} \\
&=&\left( 1-v\right) f+vh.
\end{eqnarray*}%
Hence, $\left\{ \left( 1-v\right) f+vh\mid f\text{ and }h\text{ are
idempotents in }\mathbb{F}_{p}\left[ x\right] /\left( x^{n}-1\right)
\right\} $ is the set of all idempotents in $R_{p,n}$.
\end{proof}

In the following theorems $R$ is a finite commutative ring with identity:

\begin{theorem}
\cite{Huffman}\cite{Taeri} Let $f,g$ be idempotents of $R\left[ x\right]
/\left( x^{n}-1\right) $ and let $C_{1}=\left\langle f\right\rangle $, $%
C_{2}=\left\langle g\right\rangle $ be cyclic codes over $R$. Then $%
C_{1}\cap C_{2}$ and $C_{1}+C_{2}$ have idempotent generators $fg$ and $%
f+g-fg$, respectively.
\end{theorem}

\begin{theorem}
\cite{Huffman}\cite{Taeri} Let $f\left( x\right) $ be the idempotent
generator of an $R$-cyclic code $C$. Then $1-f\left( x^{-1}\right) $ is the
idempotent generator of the dual code $C^{\perp }$.
\end{theorem}

The extended code of a code $C$ over $\mathbb{F}_{p}+v\mathbb{F}_{p}$ will
be denoted by $\overline{C}$, which is the code obtained by adding a
specific column to the generator matrix of $C$.

For the rest of this work, $q$ is an odd prime such that $p$ is a quadratic
residue modulo $q$, we set $e_{1}(x)=\sum\limits_{i\in Q_{q}}x^{i}$ and $%
e_{2}(x)=\sum\limits_{i\in N_{q}}x^{i}$, where $Q_{q}$ denotes the set of
quadratic residues modulo $q$ and $N_{q}$ denotes the set of quadratic
non-residues modulo $q$.

Let $a$ be a non-zero element of $\mathbb{F}_{q}$, the map $\mu _{a}:\mathbb{%
F}_{q}\rightarrow \mathbb{F}_{q}$ is defined as $\mu _{a}\left( i\right) =ai%
\pmod{q}$. This map acts on polynomials as
\begin{equation*}
\mu _{a}\left( \sum_{i}x^{i}\right) =\sum_{i}x^{\mu _{a}(i)}.
\end{equation*}%
It is easy to see that $\mu _{a}\left( fg\right) =\mu _{a}\left( f\right)
\mu _{a}\left( g\right) $ for polynomials $f$ and $g$ in $R_{p,q}$.

\subsection{Case I}

If $p$ is an odd prime,%
\begin{eqnarray*}
\varphi &:&\mathbb{F}_{p}+v\mathbb{F}_{p}\rightarrow \mathbb{F}_{p}^{2} \\
a+bv &\mapsto &\left( -b,2a+b\right) \text{.}
\end{eqnarray*}%
was defined as the Gray map in \cite{Zhu2}. The Lee weight of an element in $%
\mathbb{F}_{p}+v\mathbb{F}_{p}$ is defined as the Hamming weight of its Gray
image; in other words%
\begin{equation*}
w_{L}\left( a+bv\right) =\left\{
\begin{array}{l}
0,\text{if }a=0,b=0 \\
1,\text{if }a\neq 0,b=0 \\
1,\text{if }b\neq 0,2a+b\equiv 0\pmod{p} \\
2,\text{if }b\neq 0,2a+b\neq 0\pmod{p}%
\end{array}%
\right.
\end{equation*}%
The Gray map $\varphi $ is extended to $\left( \mathbb{F}_{p}+v\mathbb{F}%
_{p}\right) ^{n}$ componentwise, naturally.

\begin{proposition}
\cite{Zhu2} The Gray map $\varphi $ is a distance-preserving map from ($%
\left( \mathbb{F}_{p}+v\mathbb{F}_{p}\right) ^{n}$, Lee distance) to ($%
\mathbb{F}_{p}^{2n}$, Hamming distance) and it is also $\mathbb{F}_{p}$%
-linear.
\end{proposition}

\begin{proposition}
\label{duality}The Gray image of a self-dual code over $\mathbb{F}_{p}+v%
\mathbb{F}_{p}$ is a $p$-ary self-dual code.
\end{proposition}

\begin{proof}
It is enough to show that the extended Gray map preserves orthogonality,
then the result follows from the size of the image. Let $\overline{a}+%
\overline{b}v$ and $\overline{c}+\overline{d}v$ where $\overline{a},%
\overline{b},\overline{c}$ and $\overline{d}\in \mathbb{F}_{p}^{n}$  be two
codewords of length $n$ over $\mathbb{F}_{p}+v\mathbb{F}_{p}$ such that they
are orthogonal, then
\begin{eqnarray*}
\left( \overline{a}+\overline{b}v\right) \left( \overline{c}+\overline{d}%
v\right)  &=&0 \\
\overline{ac}+\left( \overline{ad+bc+bd}\right) v &=&0
\end{eqnarray*}%
Now, consider the inner product of the Gray images;
\begin{eqnarray*}
\left( \overline{-b},\overline{2a+b}\right) \left( \overline{-d},\overline{%
2c+d}\right)  &=&\overline{bd+4ac+2ad+2bc+bd} \\
&=&\overline{4ac+2\left( ad+bc+bd\right) }.
\end{eqnarray*}%
It is easily observed that if two codewords are orthogonal then so are their
Gray images.
\end{proof}

In \cite{Zhu2} it was shown that $S_{n}:=(\mathbb{F}_{p}+v\mathbb{F}%
_{p})[x]/(x^{n}-\left( 1-2v\right) )$ is a principal ideal ring. In the
following proposition we describe the cyclic codes of odd lengths in terms
of ideals in $S_{n}$:

\begin{proposition}
Let $\psi :\left( \mathbb{F}_{p}+v\mathbb{F}_{p}\right) \left[ x\right]
/\left( x^{n}-1\right) \rightarrow (\mathbb{F}_{p}+v\mathbb{F}%
_{p})[x]/(x^{n}-\left( 1-2v\right) )$ be defined as
\begin{equation*}
\psi \left( c\left( x\right) \right) =c\left( \left( 1-2v\right) x\right) .
\end{equation*}%
If $n$ is odd, then $\psi $ is a ring isomorphism from $R_{p,n}$ to $S_{n}$.
\end{proposition}

\begin{proof}
Note that $1-2v$ is a unit in $\mathbb{F}_{p}+v\mathbb{F}_{p}$ with $\left(
1-2v\right) ^{2}=1$. But this implies that if $n$ is odd, then
\begin{equation*}
\left( 1-2v\right) ^{n}=\left( 1-2v\right) .
\end{equation*}%
Now, suppose $a\left( x\right) \equiv b\left( x\right) \pmod{
x^{n}-1}$, i.e. $a\left( x\right) -b\left( x\right) =\left( x^{n}-1\right)
q\left( x\right) $ for some $q\left( x\right) \in \left( \mathbb{F}_{p}+v%
\mathbb{F}_{p}\right) \left[ x\right] $. Then%
\begin{eqnarray*}
a\left( \left( 1-2v\right) x\right) -b\left( \left( 1-2v\right) x\right)
&=&\left( \left( 1-2v\right) ^{n}x^{n}-1\right) q\left( \left( 1-2v\right)
x\right) \\
&=&\left( \left( 1-2v\right) x^{n}-\left( 1-2v\right) ^{2}\right) q\left(
\left( 1-2v\right) x\right) \\
&=&\left( 1-2v\right) \left( x^{n}-\left( 1-2v\right) \right) q\left( \left(
1-2v\right) x\right) ,
\end{eqnarray*}%
which means if $a\left( x\right) \equiv b\left( x\right) \pmod{
x^{n}-1}$, then $a\left( \left( 1-2v\right) x\right) \equiv b\left( \left(
1-2v\right) x\right) \pmod{x^{n}-\left(
1-2v\right)}$. The converse can easily be shown as well which means%
\begin{equation*}
a\left( x\right) \equiv b\left( x\right) \pmod{x^{n}-1}\Leftrightarrow
a\left( \left( 1-2v\right) x\right) \equiv b\left( \left( 1-2v\right)
x\right) \pmod{x^{n}-\left( 1-2v\right)}.
\end{equation*}%
Note that one side of the implication tells us that $\psi $ is well-defined
and the other side tells us that it is injective, but since the rings are
finite this proves that $\psi $ is an isomorphism.
\end{proof}

The following corollaries follow naturally from the proposition:

\begin{corollary}
\label{isomorphism} $I$ is an ideal of $R_{p,n}$ if and only if $\psi \left(
I\right) $ is an ideal of $S_{n}$ when $n$ is odd.
\end{corollary}

\begin{corollary}
\label{principalodd}Cyclic codes over $\mathbb{F}_{p}+v\mathbb{F}_{p}$ of
odd length are principally generated.
\end{corollary}

\subsection{Case II}

If $p=2$ then the Lee weight for $\mathbb{F}_{2}+v\mathbb{F}_{2}$ is defined
as
\begin{equation*}
w_{L}\left( 0\right) =0,w_{L}\left( 1\right) =2,w_{L}\left( 1+v\right)
=1,w_{L}\left( v\right) =1\text{,}\
\end{equation*}%
and the following Gray map is a linear isometry with respect to the Lee
weight:

\begin{eqnarray*}
\varphi &:&\mathbb{F}_{2}+v\mathbb{F}_{2}\rightarrow \mathbb{F}_{2}^{2} \\
a+bv &\mapsto &\left( a,a+b\right) \text{.}
\end{eqnarray*}

The Gray map is extended componentwise and preserves self-duality \cite%
{Dougherty1}.

In \cite{Bachoc}, Bachoc defined the following weight on $\mathbb{F}_{2}+v%
\mathbb{F}_{2}$:
\begin{equation*}
w_{B}\left( 0\right) =0,w_{B}\left( 1\right) =1,w_{B}\left( 1+v\right)
=2,w_{B}\left( v\right) =2\text{.}\
\end{equation*}

The weight of a codeword is the sum of the weights of its components. The
minimum Hamming, Lee and Bachoc weights, $d_{H},\ d_{L}$ and $d_{B}$ of $C$
are the smallest Hamming, Lee and Bachoc weights among the non-zero
codewords of $C$, respectively.

Let $x=\left( x_{1},x_{2},\ldots ,x_{n}\right) $ and $y=\left(
y_{1},y_{2},\ldots ,y_{n}\right) $ be two elements of $\left( \mathbb{F}%
_{2}+v\mathbb{F}_{2}\right) ^{n}$. The Hermitian inner product is defined as
$\left\langle x,y\right\rangle _{H}=\sum x_{i}\overline{y_{i}}$ where $%
\overline{0}=0,\ \overline{1}=1,\ \overline{v}=1+v$ and$\ \overline{1+v}=v$.
The dual code $C^{\ast }$ with respect to the Hermitian inner product of $C$
is defined as
\begin{equation*}
C^{\ast }=\left\{ x\in \left( \mathbb{F}_{2}+v\mathbb{F}_{2}\right) ^{n}\mid
\left\langle x,y\right\rangle _{H}=0\text{ for all }y\in C\right\} .
\end{equation*}%
$C$ is Hermitian self-dual if $C=C^{\ast }$.

The following theorems, taken from \cite{Zhu} characterize the structure of
cyclic codes over the ring $\mathbb{F}_2+v\mathbb{F}_2$:

\begin{theorem}
\cite{Zhu} \label{cyclic} For any cyclic code $C$ of length $n$ over $%
\mathbb{F}_{2}+v\mathbb{F}_{2}$, there is a unique polynomial $g\left(
x\right) $ such that $C=\left\langle g\left( x\right) \right\rangle $, and $%
g\left( x\right) \mid x^{n}-1$.
\end{theorem}

\begin{corollary}
\cite{Zhu} Every ideal of $R_{2,n}=(\mathbb{F}_{2}+v\mathbb{F}%
_{2})[x]/(x^{n}-1)$ is principal.
\end{corollary}

\begin{theorem}
\cite{Zhu} \label{idempotent} If $n$ is odd then every cyclic code over $%
\mathbb{F}_{2}+v\mathbb{F}_{2}$ has a unique idempotent generator.
\end{theorem}

\begin{corollary}
Any cylic code $C$ of odd length $q$ over $\mathbb{F}_{2}+v\mathbb{F}_{2}$
has a unique idempotent generator of the form $\left( 1+v\right) f+vh$ where
$f$ and $h$ are idempotents in $\mathbb{F}_{2}\left[ x\right] /\left(
x^{q}-1\right) $.
\end{corollary}

\begin{proof}
Since $q$ is odd this is an immediate consequence of Theorem \ref{idempotent}
and Lemma \ref{idem}.
\end{proof}

\section{Quadratic Residue Codes over $\mathbb{F}_{p}+v\mathbb{F}_{p}$}

Let $q$ be a prime such that $p$ is a quadratic residue in $\mathbb{F}_{q}$.
So, QR-codes of length $q$ over $\mathbb{F}_{p}$ exist. For the idempotent
generators of these codes we refer to \cite{Macwilliams}. Let $a$ and $b$ be
the idempotent generators of $\left[ q,\frac{q+1}{2}\right] $ QR-codes over $%
\mathbb{F}_{p}$ and $a^{\prime },b^{\prime }$ be the idempotent generators
of $\left[ q,\frac{q-1}{2}\right] $ QR-codes over $\mathbb{F}_{p}$.

\begin{definition}
\label{defQR} Let $q$ be a prime such that $p$ is a quadratic residue modulo
$q$. Set $Q_{1}=\left\langle \left( 1-v\right) a+vb\right\rangle $, $%
Q_{2}=\left\langle \left( 1-v\right) b+va\right\rangle $ and $\left\langle
Q_{1}^{\prime }=\left( 1-v\right) a^{\prime }+vb^{\prime }\right\rangle $, $%
Q_{2}^{\prime }=\left\langle \left( 1-v\right) b^{\prime }+va^{\prime
}\right\rangle $. These four codes are called quadratic residue codes over $%
\mathbb{F}_{p}+v\mathbb{F}_{p}$ of length $q$.
\end{definition}

As in the case of QR-codes over finite fields, the properties of QR-codes
over $\mathbb{F}_{p}+v\mathbb{F}_{p}$ differ for the cases $q\equiv 3\pmod{4}
$ and $q\equiv 1\pmod{4}$.

\subsection{Case I}

If $q\equiv 3\pmod{4}$ then the codes have the following properties.

\begin{theorem}
\label{bQR} With the notation as in the Definition \ref{defQR}, the
following hold for $(\mathbb{F}_{p}+v\mathbb{F}_{p})$-QR codes:

\item[a)] $Q_{1}$ and $Q_{1}^{\prime }$ are equivalent to $Q_{2}$ and $%
Q_{2}^{\prime }$, respectively;

\item[b)] $Q_{1}\cap Q_{2}=\left\langle h\right\rangle $ and $%
Q_{1}+Q_{2}=R_{p,q}$ where $h=1+e_{1}+e_{2}$ is the polynomial corresponding
to the all one vector of length $q$;

\item[c)] $\left\vert Q_{1}\right\vert =p^{\left( q+1\right) }=\left\vert
Q_{2}\right\vert ;$

\item[d)] $Q_{1}=Q_{1}^{\prime }+\left\langle h\right\rangle ,\
Q_{2}=Q_{2}^{\prime }+\left\langle h\right\rangle ;$

\item[e)] $\left\vert Q_{1}^{\prime }\right\vert =p^{\left( q-1\right)
}=\left\vert Q_{2}^{\prime }\right\vert ;$

\item[f)] $Q_{1}^{\prime }$ and $Q_{2}^{\prime }$ are self-orthogonal and $%
Q_{1}^{\bot }=Q_{1}^{\prime }$ and $Q_{2}^{\bot }=Q_{2}^{\prime };$

\item[g)] $Q_{1}^{\prime }\cap Q_{2}^{\prime }=\left\{ 0\right\} $ and $%
Q_{1}^{\prime }+Q_{2}^{\prime }=\left\langle 1+h\right\rangle .$
\end{theorem}

\begin{proof}

\item[a)] Let $n\in N_{q}$ then $\mu _{n}a=b$ and $\mu _{n}a^{\prime
}=b^{\prime }$ therefore
\begin{equation*}
\mu _{n}\left[ \left( 1-v\right) a+vb\right] =\left( 1-v\right) b+va
\end{equation*}%
so $Q_{1}$ and $Q_{2}$ are equivalent. Similarly,
\begin{equation*}
\mu _{n}\left[ \left( 1-v\right) a^{\prime }+vb^{\prime }\right] =\left(
1-v\right) b^{\prime }+va^{\prime }
\end{equation*}%
which implies $Q_{1}^{\prime }$ and $Q_{2}^{\prime }$ are equivalent.

\item[b)] $Q_{1}\cap Q_{2}$ has idempotent generator%
\begin{equation*}
\left[ \left( 1-v\right) a+vb\right] \left[ \left( 1-v\right) b+va\right]
=\left( 1-v\right) ab+vab=ab=h
\end{equation*}%
and $Q_{1}+Q_{2}$ has idempotent generator $\left( 1-v\right) a+vb+\left(
1-v\right) b+va-ab=a+b-ab=1$.

\item[c)] By a) and b) we have
\begin{equation*}
\left( p^{2}\right) ^{q}=\left\vert Q_{1}+Q_{2}\right\vert =\frac{\left\vert
Q_{1}\right\vert \left\vert Q_{2}\right\vert }{\left\vert Q_{1}\cap
Q_{2}\right\vert }=\frac{\left\vert Q_{1}\right\vert ^{2}}{p^{2}},
\end{equation*}%
so $\left\vert Q_{1}\right\vert =\left\vert Q_{2}\right\vert =p^{\left(
q+1\right) }$.

\item[d)] $Q_{1}^{\prime }+\left\langle h\right\rangle $ has idempotent
generator
\begin{eqnarray*}
&&\left( 1-v\right) a^{\prime }+vb^{\prime }+h-\left[ \left( 1-v\right)
a^{\prime }+vb^{\prime }\right] h \\
&=&\left( 1-v\right) \left[ a^{\prime }+h-a^{\prime }h\right] +v\left[
b^{\prime }+h-b^{\prime }h\right] =\left( 1-v\right) a+vb
\end{eqnarray*}
Hence, $Q_{1}^{\prime }+\left\langle h\right\rangle =Q_{1}$. Similarly, $%
Q_{2}=Q_{2}^{\prime }+\left\langle h\right\rangle $.

\item[e)] $p^{\left( q+1\right) }=\left\vert Q_{1}\right\vert =\left\vert
Q_{1}^{\prime }+\left\langle h\right\rangle \right\vert =\left\vert
Q_{1}^{\prime }\right\vert \left\vert \left\langle h\right\rangle
\right\vert =p^{2}\left\vert Q_{1}^{\prime }\right\vert $. Thus $\left\vert
Q_{1}^{\prime }\right\vert =p^{\left( q-1\right) }$.

\item[f)] $Q_{1}^{\bot }$ has idempotent generator $1-\left( \left(
1-v\right) a\left( x^{-1}\right) +vb\left( x^{-1}\right) \right) $ and since
$-1\in N_{q},\ a\left( x^{-1}\right) =b$ and $b\left( x^{-1}\right) =a$. So
we obtain%
\begin{eqnarray*}
1-\left( \left( 1-v\right) a\left( x^{-1}\right) +vb\left( x^{-1}\right)
\right) &=&1-v+\left( 1-v\right) a\left( x^{-1}\right) +v-vb\left(
x^{-1}\right) \\
&=&\left( 1-v\right) \left[ 1-a\left( x^{-1}\right) \right] +v\left[
1-b\left( x^{-1}\right) \right] \\
&=&\left( 1-v\right) a^{\prime }+vb^{\prime }\text{.}
\end{eqnarray*}%
which implies $Q_{1}^{\bot }=Q_{1}^{\prime }$. Similarly, $Q_{2}^{\bot
}=Q_{2}^{\prime }$. By d) $Q_{1}^{\prime }\subset Q_{1}$ and $Q_{2}^{\prime
}\subset Q_{2}$ it follows that they are self orthogonal.

\item[g)] $Q_{1}^{\prime }\cap Q_{2}^{\prime }$ has idempotent generator
\begin{eqnarray*}
&&\left[ \left( 1-v\right) a^{\prime }+vb^{\prime }\right] \left[ \left(
1-v\right) b^{\prime }+va^{\prime }\right] \\
&=&\left( 1-v\right) a^{\prime }b^{\prime }+va^{\prime }b^{\prime } \\
&=&a^{\prime }b^{\prime }=0.
\end{eqnarray*}%
$Q_{1}^{\prime }+Q_{2}^{\prime }$ has idempotent generator $\left(
1-v\right) a^{\prime }+vb^{\prime }+\left( 1-v\right) b^{\prime }+va^{\prime
}-0=a^{\prime }+b^{\prime }=1+h$.
\end{proof}

In the following, the extended QR-codes over $\mathbb{F}_{p}+v\mathbb{F}_{p}$
are formed by adding the same columns that are used to extend QR-codes over $%
\mathbb{F}_{p}$.

\begin{theorem}
\label{43}Suppose $q\equiv 3\pmod{4}$ and $Q_{1},Q_{2}$ are $\mathbb{F}_{p}+v%
\mathbb{F}_{p}$-QR codes as given in Definition \ref{defQR}. Then $\overline{%
Q_{1}}$ and $\overline{Q_{2}}$ are self-dual.
\end{theorem}

\begin{proof}
If $p\equiv 1\pmod{4}$ then both $-1$ and $q$ are quadratic residues in $%
\mathbb{F}_{p}$. If $p\equiv 3\pmod{4}$ then both $-1$ and $q$ are
non-residues in $\mathbb{F}_{p}$. These follow easily from the quadratic
reciprocity law. So, in any case $-q$ is a quadratic residue in $\mathbb{F}%
_{p}$, i.e. there exists $r\in \mathbb{F}_{p}$ such that $r^{2}=-q$. By
Theorem \ref{bQR}, $Q_{1}=Q_{1}^{\prime }+\left\langle h\right\rangle $. Let
$\overline{Q_{1}}$ be the code generated by%
\begin{equation*}
\begin{array}{cccccc}
\infty & 0 & 1 & 2 & \cdots & q-1%
\end{array}%
\end{equation*}%
\begin{equation}
\left(
\begin{array}{cccccc}
0 &  &  &  &  &  \\
0 &  &  & G_{1}^{\prime } &  &  \\
\vdots &  &  &  &  &  \\
r & 1 & 1 & 1 & \cdots & 1%
\end{array}%
\right)  \label{gen}
\end{equation}%
where $G_{1}^{\prime }$ is a generator matrix for $Q_{1}^{\prime }$ and $r$
is an element with $r^{2}\equiv -q\pmod{p}$. Note that the all $1$ vector is
in $Q_{1}$. Hence, since $Q_{1}^{\bot }=Q_{1}^{\prime }$, the last row is
orthogonal to the rows above; moreover it is orthogonal to itself, making
the code $\overline{Q_{1}}$ self-orthogonal. By comparing $\left\vert
\overline{Q_{1}}\right\vert $ and $\left\vert \overline{Q_{1}}^{\bot
}\right\vert $ it follows that $\overline{Q_{1}}$ is self-dual. Similarly, $%
\overline{Q_{2}}$ is self-dual.
\end{proof}

\subsection{\protect\bigskip Case II}

If $q\equiv 1\pmod{4}$ then the codes have the following properties.

\begin{theorem}
\label{bQR41} With the notation as in the Definition \ref{defQR}, the
following hold for $(\mathbb{F}_{p}+v\mathbb{F}_{p})$-QR codes:

\item[a)] $Q_{1}$ and $Q_{1}^{\prime }$ are equivalent to $Q_{2}$ and $%
Q_{2}^{\prime }$, respectively;

\item[b)] $Q_{1}\cap Q_{2}=\left\langle h\right\rangle $ and $%
Q_{1}+Q_{2}=R_{p,q}$ where $h=1+e_{1}+e_{2}$;

\item[c)] $\left\vert Q_{1}\right\vert =p^{\left( q+1\right) }=\left\vert
Q_{2}\right\vert ;$

\item[d)] $Q_{1}=Q_{1}^{\prime }+\left\langle h\right\rangle ,\
Q_{2}=Q_{2}^{\prime }+\left\langle h\right\rangle ;$

\item[e)] $\left\vert Q_{1}^{\prime }\right\vert =p^{\left( q-1\right)
}=\left\vert Q_{2}^{\prime }\right\vert ;$

\item[f)] $Q_{1}^{\bot }=Q_{2}^{\prime }$ and $Q_{2}^{\bot }=Q_{1}^{\prime
}; $

\item[g)] $Q_{1}^{\prime }\cap Q_{2}^{\prime }=\left\{ 0\right\} $ and $%
Q_{1}^{\prime }+Q_{2}^{\prime }=\left\langle 1+h\right\rangle .$
\end{theorem}

\begin{proof}
The proof is similar to the proof of Theorem \ref{bQR} and hence is omitted
here.
\end{proof}

We finish this section with the following theorem describing the duality
relation between the extended quadratic residue codes:

\begin{theorem}
\label{41}Suppose $q\equiv 1\pmod{4}$ and $Q_{1},Q_{2}$ are $\mathbb{F}_{p}+v%
\mathbb{F}_{p}$-QR codes as given in Definiton \ref{defQR}. Then the dual of
$\overline{Q_{1}}$ is $\overline{Q_{2}}$ and the dual of $\overline{Q_{2}}$
is $\overline{Q_{1}}$.
\end{theorem}

\begin{proof}
Since, $Q_{i}=Q_{i}^{\prime }+\left\langle h\right\rangle $ for $i=1,2$, $%
\overline{Q_{i}}$ can be defined as the code generated by the following
matrix $\overline{G_{i}}$;%
\begin{equation*}
\begin{array}{cccccc}
\infty & 0 & 1 & 2 & \cdots & q-1%
\end{array}%
\end{equation*}%
\begin{equation*}
\left(
\begin{array}{cccccc}
0 &  &  &  &  &  \\
0 &  &  & G_{i}^{\prime } &  &  \\
\vdots &  &  &  &  &  \\
g_{i} & 1 & 1 & 1 & \cdots & 1%
\end{array}%
\right)
\end{equation*}%
where $G_{i}^{\prime }$ is a generator matrix of $Q_{i}^{\prime }$ and $%
g_{1}=1,g_{2}=-q$. By Theorem \ref{bQR41} $Q_{2}^{\bot }=Q_{1}^{\prime }$
and $Q_{1}^{\bot }=Q_{2}^{\prime }$. Since the all $1$ vector is in both $%
Q_{1}$ and $Q_{2}$, this implies that it is orthogonal to both $%
G_{1}^{\prime }$ and $G_{2}^{\prime }$. Moreover, the last row of $\overline{%
G_{1}}$ is orthogonal to the last row of $\overline{G_{2}}$. It follows that
the codes are orthogonal to each other. By comparing their orders we see
that $\overline{Q_{1}}^{\bot }=\overline{Q_{2}}$.
\end{proof}

Recall that $Q_{1}$ and $Q_{2}$ are equivalent and thus have the same weight
enumerator. Since $w_{L}\left( 1\right) =w_{L}\left( -q\right) $, $\overline{%
Q_{1}}$ and $\overline{Q_{2}}$ have the same weight enumerator. Thus, by
theorems \ref{43} and \ref{41} and proposition \ref{duality} we obtain the
following corollary:

\begin{corollary}
The Gray images of the extended quadratic residue codes over $\mathbb{F}%
_{p}+v\mathbb{F}_{p}$ are self-dual codes if $q\equiv 3\pmod{4}$ and
formally self-dual codes if $q\equiv 1\pmod{4}$.
\end{corollary}

\section{Quadratic residue codes over $\mathbb{F}_{2}+v\mathbb{F}_{2}$,
extended quadratic residue codes and binary images}

In this section, we define extended quadratic residue codes over $\mathbb{F}%
_{2}+v\mathbb{F}_{2}$. Further, we provide two optimal Hermitian self-dual
codes as applications to the main theorems.

\begin{proposition}
\cite{Harada} Let $d_{H}$ and $d_{L}$ be the minimum Hamming and Lee weights
of $C=\left( 1+v\right) C_{1}\oplus vC_{2}$, respectively. Then%
\begin{equation*}
d_{H}\left( C\right) =d_{L}\left( C\right) =\min \left\{ d\left(
C_{1}\right) ,d\left( C_{2}\right) \right\}
\end{equation*}%
where $d\left( C_{i}\right) $ denotes the minimum weight of the binary code $%
C_{i}$.
\end{proposition}

A self-dual code is called Type IV if all the Hamming weights are even, a
binary code is called even if all the weights are even.

\begin{proposition}
\cite{Bachoc} \cite{Dougherty1} If $C=\left( 1+v\right) C_{1}\oplus vC_{2}$
then $C$ is Euclidean self-dual if and only if $C_{1}$ and $C_{2}$ are
binary self-dual codes. $C=\left( 1+v\right) C_{1}\oplus vC_{2}$ is
Euclidean Type IV self-dual if and only if $C_{1}=C_{2}$.
\end{proposition}

\begin{proposition}
\cite{Bachoc} \cite{Dougherty1} \label{H}If $C=\left( 1+v\right) C_{1}\oplus
vC_{2}$ then $C$ is Hermitian self-dual if and only if $C_{1}=C_{2}^{\perp }$%
. $C=\left( 1+v\right) C_{1}\oplus vC_{1}^{\perp }$ is Hermitian Type IV
self-dual if and only if $C_{1}$ and $C_{1}^{\perp }$ are even codes.
\end{proposition}

The following gives an upper bound on the Bachoc distances of Hermitian
self-dual codes.

\begin{theorem}
\cite{Bachoc} Let $C$ be a Hermitian self-dual code of length $n$ over $%
\mathbb{F}_{2}\times \mathbb{F}_{2}$ then $w_{B}\left( C\right) \leq 2\left( %
\left[ n/3\right] +1\right) $.
\end{theorem}

Codes that meet this bound are called extremal codes and Bachoc has shown
that they correspond to extremal modular lattices.

\subsection{Case I}

If $q=8r-1$ then $e_{1}$ and $e_{2}$ are generating idempotents of $\left[ q,%
\frac{q+1}{2}\right] $ binary quadratic residue codes and $e_{1}e_{2}=h$.

In this case, the $(\mathbb{F}_{2}+v\mathbb{F}_{2})$-QR codes turn out to
have the following form.

\begin{definition}
\label{defQR1} If $q=8r-1$ let $Q_{1}=\left\langle \left( 1+v\right)
e_{1}+ve_{2}\right\rangle $, $Q_{2}=\left\langle \left( 1+v\right)
e_{2}+ve_{1}\right\rangle $ and $Q_{1}^{\prime }=\left\langle \left(
1+v\right) \left( 1+e_{2}\right) +v\left( 1+e_{1}\right) \right\rangle $, $%
Q_{2}^{\prime }=\left\langle \left( 1+v\right) \left( 1+e_{1}\right)
+v\left( 1+e_{2}\right) \right\rangle $. These four codes are called
quadratic residue codes over $\mathbb{F}_{2}+v\mathbb{F}_{2}$ of length $q$.
\end{definition}

\begin{example}
For $q=7$ the QR-code $Q_{1}$ is generated by the idempotent $e=\left(
1+v\right) \left( x+x^{2}+x^{4}\right) +v\left( x^{3}+x^{5}+x^{6}\right) $
in $R_{2,7}$. As the extended QR-code $\overline{Q_{1}}$ we get the
self-dual code which is generated by%
\begin{equation*}
\left(
\begin{array}{cccccccc}
0 & 1 & 1+v & 1+v & v & 1+v & v & v \\
0 & v & 1 & 1+v & 1+v & v & 1+v & v \\
0 & v & v & 1 & 1+v & 1+v & v & 1+v \\
1 & 1 & 1 & 1 & 1 & 1 & 1 & 1%
\end{array}%
\right)
\end{equation*}%
with $d_{L}\left( \overline{Q_{1}}\right) =4$ and $d_{B}\left( \overline{%
Q_{1}}\right) =7$. So its Gray image corresponds to a $\left[ 16,8,4\right] $
optimal self-dual binary code. $\overline{Q_{1}}$ has Lee weight enumerator
\begin{equation*}
1+28z^{4}+198z^{8}+28z^{12}+z^{16},
\end{equation*}%
Hamming weight enumerator
\begin{equation*}
1+28z^{4}+56z^{5}+84z^{6}+56z^{7}+31z^{8},
\end{equation*}%
and Bachoc weight enumerator
\begin{equation*}
1+56z^{7}+29z^{8}+84z^{10}+28z^{12}+56z^{13}+2z^{16}\text{.}
\end{equation*}
\end{example}

\begin{example}
For $q=23$, the extended quadratic residue code $\overline{Q_{1}}$ is a
self-dual code with $d_{B}\left( \overline{Q_{1}}\right) =14$, $d_{H}\left(
\overline{Q_{1}}\right) =8=d_{L}\left( \overline{Q_{1}}\right) $ and Lee
weight enumerator%
\begin{eqnarray*}
&&1+1518z^{8}+5152z^{12}+577599z^{16}+3910368z^{20}+7787940z^{24} \\
&&+3910368z^{28}+577599z^{32}+5152z^{36}+1518z^{40}+z^{48}.
\end{eqnarray*}
\end{example}

\subsection{Case II}

If $q=8r+1$ then $e_{1}$ and $e_{2}$ are generating idempotents of $\left[ q,%
\frac{q-1}{2}\right] $ binary quadratic residue codes and $e_{1}e_{2}=0$.

In this case, the $(\mathbb{F}_{2}+v\mathbb{F}_{2})$-QR codes turn out to
have the following form.

\begin{definition}
\label{defQR2} If $q=8r+1$ let $Q_{1}=\left\langle \left( 1+v\right) \left(
1+e_{1}\right) +v\left( 1+e_{2}\right) \right\rangle $, $Q_{2}=\left\langle
\left( 1+v\right) \left( 1+e_{2}\right) +v\left( 1+e_{1}\right)
\right\rangle $ and $Q_{1}^{\prime }=\left\langle \left( 1+v\right)
e_{2}+ve_{1}\right\rangle $, $Q_{2}^{\prime }=\left\langle \left( 1+v\right)
e_{1}+ve_{2}\right\rangle $. These four codes are called quadratic residue
codes over $\mathbb{F}_{2}+v\mathbb{F}_{2}$ of length $q$.
\end{definition}

In the next theorem we introduce a family of Hermitian self-dual codes by
using quadratic residue codes:

\begin{theorem}
\label{Hermitian}Suppose $q=8r+1$ and $Q_{1}^{\prime },Q_{2}^{\prime }$ are $%
\mathbb{F}_{2}+v\mathbb{F}_{2}$-QR codes as given in Definition \ref{defQR2}%
. Then $Q_{1}^{\prime }+\left\langle \boldsymbol{v}\right\rangle $, $%
Q_{2}^{\prime }+\left\langle \boldsymbol{v}\right\rangle $, $Q_{1}^{\prime
}+\left\langle \boldsymbol{1+v}\right\rangle $ and $Q_{2}^{\prime
}+\left\langle \boldsymbol{1+v}\right\rangle $ are Hermitian self-dual codes
of length $q$ where $\boldsymbol{v}$ denotes the polynomial $vh$ which
corresponds to the all-$v$ vector of length $q$ and $\boldsymbol{1+v}$
denotes the polynomial $\left( 1+v\right) h$ which corresponds to the all-$%
\left( 1+v\right) $ vector.
\end{theorem}

\begin{proof}
$Q_{1}^{\prime }+\left\langle \boldsymbol{v}\right\rangle $ has idempotent
generator
\begin{eqnarray*}
&&\left( 1+v\right) e_{2}+ve_{1}+\boldsymbol{v}-\left( \left( 1+v\right)
e_{2}+ve_{1}\right) \boldsymbol{v} \\
&=&e_{2}+ve_{2}+ve_{1}+v+ve_{1}+ve_{2}-ve_{1}v\left( 1+e_{1}+e_{2}\right) \\
&=&e_{2}+v-v^{2}e_{1}-v^{2}e_{1}^{2}-v^{2}e_{1}e_{2} \\
&=&e_{2}+v-ve_{1}-ve_{1}-0 \\
&=&e_{2}+v \\
&=&\left( 1+v\right) e_{2}+v\left( 1+e_{2}\right)
\end{eqnarray*}%
So, $Q_{1}^{\prime }+\left\langle \boldsymbol{v}\right\rangle =\left(
1+v\right) C_{1}\oplus vC_{2}$ where $C_{1}=\left\langle e_{2}\right\rangle $
and $C_{2}=\left\langle 1+e_{2}\right\rangle $. $C_{2}^{\perp }$ has
idempotent generator $1-\left( 1+e_{2}\left( x^{-1}\right) \right)
=e_{2}\left( x^{-1}\right) =e_{2}$ since $-1\in Q_{q}$, it follows that $%
C_{1}=C_{2}^{\perp }$. Hence, by Proposition \ref{H}. $Q_{1}^{\prime
}+\left\langle \boldsymbol{v}\right\rangle $ is Hermitian self-dual.
Similarly, $Q_{2}^{\prime }+\left\langle \boldsymbol{v}\right\rangle $ is
Hermitian self-dual.

$Q_{1}^{\prime }+\left\langle \boldsymbol{1+v}\right\rangle $ has idempotent
generator%
\begin{eqnarray*}
&&\left( 1+v\right) e_{2}+ve_{1}+\boldsymbol{1+v}-\left( \left( 1+v\right)
e_{2}+ve_{1}\right) \left( \boldsymbol{1+v}\right) \\
&=&1+v+e_{1}-\left( 1+v\right) e_{2}\left( 1+e_{1}+e_{2}\right) \\
&=&1+v+e_{1}-\left( 1+v\right) \left( e_{2}+0+e_{2}\right) \\
&=&1+v+e_{1} \\
&=&\left( 1+v\right) \left( 1+e_{1}\right) +ve_{1}.
\end{eqnarray*}%
So, $Q_{1}^{\prime }+\left\langle \boldsymbol{1+v}\right\rangle =\left(
1+v\right) C_{1}\oplus vC_{2}$ where $C_{1}=\left\langle
1+e_{1}\right\rangle $ and $C_{2}=\left\langle e_{1}\right\rangle $. $%
C_{2}^{\perp }$ has idempotent generator $1-e_{1}\left( x^{-1}\right)
=1-e_{1}=1+e_{1}$. Hence, by Proposition \ref{H} $Q_{1}^{\prime
}+\left\langle \boldsymbol{1+v}\right\rangle $ is Hermitian self-dual.
Similarly, $Q_{2}^{\prime }+\left\langle \boldsymbol{1+v}\right\rangle $ is
Hermitian self-dual.
\end{proof}

\begin{example}
\label{optimal17}For $q=17$, the code $Q_{1}^{\prime }+\left\langle
\boldsymbol{v}\right\rangle $ is the unique optimal Hermitian self-dual code
of length $17$ with $d_{B}\left( Q_{1}^{\prime }+\left\langle v\right\rangle
\right) =10$ and Bachoc weight enumerator
\begin{eqnarray*}
&&1+187z^{10}+1156z^{12}+2924z^{14}+10030z^{16}+18513z^{18}+27744z^{20} \\
&&+29954z^{22}+23188z^{24}+12019z^{26}+850z^{30}+85z^{32}+z^{34}\text{.}
\end{eqnarray*}
\end{example}

\begin{theorem}
Suppose $q=8r+1$ and $Q_{1},Q_{2}$ are $\mathbb{F}_{2}+v\mathbb{F}_{2}$-QR
codes as given in Definition \ref{defQR2}. Then $\overline{Q_{1}}$ and $%
\overline{Q_{2}}$ are Hermitian self-dual codes.
\end{theorem}

\begin{proof}
By Theorem \ref{bQR41} $Q_{1}=Q_{1}^{\prime }+\left\langle h\right\rangle ,$
and $\overline{Q_{1}}$ has the $\frac{q+1}{2}\times \left( q+1\right) $
generator matrix (\ref{gen}). By Theorem \ref{Hermitian}. $Q_{1}^{\prime
}+\left\langle \boldsymbol{v}\right\rangle $ is Hermitian self-dual and it
is easily seen that any row of this generator matrix is orthogonal to all-$1$
vector since $\left\vert Q_{q}\right\vert =\left\vert N_{q}\right\vert =%
\frac{q-1}{2}=4r$ and all-$1$ vector is orthogonal to itself with respect to
Hermitian inner product. Therefore $\overline{Q_{1}}$ is Hermitian
self-dual. Similarly, $\overline{Q_{2}}$ is Hermitian.
\end{proof}

\begin{example}
\label{optimal18}For $q=17$, the extended quadratic residue code $\overline{%
Q_{1}}$ is the unique optimal Hermitian self-dual code of length $18$ with $%
d_{B}\left( \overline{Q_{1}}\right) =12$ and Bachoc weight enumerator%
\begin{eqnarray*}
&&1+1734z^{12}+1836z^{14}+13158z^{16}+23869z^{18}+46818z^{20}+55080z^{22} \\
&&+57324z^{24}+37026z^{26}+18054z^{28}+6324z^{30}+756z^{32}+153z^{34}+2z^{36}
\end{eqnarray*}%
$d_{H}\left( \overline{Q_{1}}\right) =6=d_{L}\left( \overline{Q_{1}}\right) $
and $\overline{Q_{1}}$ is an optimal Hermitian Type IV self-dual code of
length $18$ as given in \cite{Harada}.
\end{example}

In \cite{betsumiya} it is proven that there are no extremal codes with
respect to the Bachoc weight for the lengths greater than $10$, a self-dual
code is called optimal if it has the best possible distance. Betsumiya et.
al. obtained unique optimal Hermitian self-dual codes for lengths $17$ and $%
18$ which are obtained by quadratic residue codes in a different way in the
examples \ref{optimal18} and \ref{optimal17}.

\section{Examples of quadratic residue codes over $\mathbb{F}_{p}+v\mathbb{F}%
_{p}$, for odd prime $p$}

In this section, we investigate the examples of QR-codes over $\mathbb{F}%
_{p}+v\mathbb{F}_{p}$ for odd prime $p$. We obtain some extremal and optimal
self-dual codes over $\mathbb{F}_{p}$ as Gray images of the extended
QR-codes over $\mathbb{F}_{p}+v\mathbb{F}_{p}$. By $QR_{p}\left( q\right) $
we denote the $\left[ q,\frac{q+1}{2}\right] $-quadratic residue code of
length $q$ over $\mathbb{F}_{p}+v\mathbb{F}_{p}$. As usual the notation $%
\overline{QR_{p}\left( q\right) }$ is used for the extended QR-code.

\begin{example}
For $p=3$ and $q=11$ we get the $\left( \mathbb{F}_{3}+v\mathbb{F}%
_{3}\right) $-QR code $QR_{3}\left( 11\right) $ which is generated by the
matrix%
\begin{equation*}
\left(
\begin{array}{ccccccccccc}
0 & a & v & a & a & a & v & v & v & a & v \\
v & 0 & a & v & a & a & a & v & v & v & a \\
a & v & 0 & a & v & a & a & a & v & v & v \\
v & a & v & 0 & a & v & a & a & a & v & v \\
v & v & a & v & 0 & a & v & a & a & a & v \\
v & v & v & a & v & 0 & a & v & a & a & a%
\end{array}%
\right)
\end{equation*}%
where $a=1+2v$. It has minimum Lee distance $7$, so it corresponds to a $%
\left[ 22,12,7\right] _{3}$ optimal code. Moreover, $\overline{QR_{3}\left(
11\right) }$ has minimum Lee distance $9$ and it is self-dual so its Gray
image is an extremal code with parameters $\left[ 24,12,9\right] _{3}$. Note
that this coincides with the extended QR-code over $\mathbb{F}_{3}$.
\end{example}

\begin{example}
For $p=3$ and $q=23$, the Gray image of $QR_{3}\left( 23\right) $
corresponds to an optimal $\left[ 46,24,13\right] _{3}$ code and the image
of $\overline{QR_{3}\left( 23\right) }$ corresponds to an extremal self-dual
$\left[ 48,24,15\right] _{3}$-code which coincides with the extended QR-code
over $\mathbb{F}_{3}$.
\end{example}

\begin{example}
For $p=5$ and $q=11$, $QR_{5}\left( 11\right) $ is generated by the
idempotent $\left( 1+4v\right) \left( 1+2e_{1}+4e_{2}\right) +v\left(
1+2e_{2}+4e_{1}\right) $ and its Gray image is a $\left[ 22,12,7\right] _{5}$%
-code which is optimal. The image of $\overline{QR_{5}\left( 11\right) }$ is
the unique optimal self-dual $\left[ 24,12,9\right] _{5}$ code with weight
enumerator;%
\begin{equation*}
1+1056z^{9}+11018z^{10}+36960z^{11}+212352z^{12}+\cdots .
\end{equation*}
\end{example}

\begin{example}
For $p=5$ and $q=19$, $QR_{5}\left( 19\right) $ is generated by the
idempotent $\left( 1+4v\right) 4e_{1}+v4e_{2}$ and its image is a $\left[
38,20,11\right] _{5}$-code which is optimal. $\overline{QR_{5}\left(
19\right) }$ is self-dual and its Gray image corresponds to an optimal
self-dual code with parameters $\left[ 40,20,13\right] _{5}$.
\end{example}

\begin{example}
For $p=7$ and $q=19$, $QR_{7}\left( 19\right) $ has generating idempotent $%
\left( 1+6v\right) \left( 2+4e_{1}+6e_{2}\right) +v\left(
2+4e_{2}+6e_{1}\right) $. We obtain a self-dual code of parameters $\left[
40,20,13\right] _{7}$ as the Gray image of $\overline{QR_{7}\left( 19\right)
}$.
\end{example}

We finish this section by combining the codes obtained in the following
tables:

\begin{center}
Table 1: QR codes for $p=3,5$ and $7$

\begin{tabular}{|l|l|l|}
\hline
& The code over $\mathbb{F}_{p}+v\mathbb{F}_{p}$ & Gray image; over $\mathbb{%
F}_{p}$ \\ \cline{2-3}
$QR_{3}\left( 11\right) $ & $\left( 11,9^{6},7\right) $ & $%
[22,12,7]_{3}^{\ast }$ \\ \cline{2-3}
$\overline{QR_{3}\left( 11\right) }$ & $\left( 12,9^{6},9\right) $ & $%
[24,12,9]_{3}^{\ast }$ extremal self-dual \\ \cline{2-3}
$QR_{3}\left( 13\right) $ & $\left( 13,9^{7},7\right) $ & $%
[26,14,7]_{3}^{\ast }$ \\ \cline{2-3}
$\overline{QR_{3}\left( 13\right) }$ & $\left( 14,9^{7},8\right) $ & $%
[28,14,8]_{3}$ formally self-dual \\ \cline{2-3}
$QR_{3}\left( 23\right) $ & $\left( 23,9^{12},13\right) $ & $%
[46,24,13]_{3}^{\ast }$ \\ \cline{2-3}
$\overline{QR_{3}\left( 23\right) }$ & $\left( 24,9^{12},15\right) $ & $%
[48,24,15]_{3}^{\ast }$ extremal self-dual \\ \cline{2-3}
$QR_{3}\left( 37\right) $ & $\left( 37,9^{19},14\right) $ & $[74,38,14]_{3}$
\\ \cline{2-3}
$\overline{QR_{3}\left( 37\right) }$ & $\left( 38,9^{19},16\right) $ & $%
[76,38,16]_{3}$ formally self-dual \\ \cline{2-3}
$QR_{5}\left( 11\right) $ & $\left( 11,25^{6},7\right) $ & $[22,12,7]_{5}$
\\ \cline{2-3}
$\overline{QR_{5}\left( 11\right) }$ & $\left( 12,25^{6},9\right) $ & $%
[24,12,9]_{5}^{\ast }$ \ optimal self-dual \\ \cline{2-3}
$QR_{5}\left( 19\right) $ & $\left( 19,25^{10},11\right) $ & $%
[38,20,11]_{5}^{\ast }$ \\ \cline{2-3}
$\overline{QR_{5}\left( 19\right) }$ & $\left( 20,25^{10},13\right) $ & $%
[40,20,13]_{5}^{\ast }$ optimal self-dual \\ \cline{2-3}
$QR_{5}\left( 29\right) $ & $\left( 29,25^{15},13\right) $ & $[58,30,13]_{5}$
\\ \cline{2-3}
$\overline{QR_{5}\left( 29\right) }$ & $\left( 30,25^{15},14\right) $ & $%
[60,30,14]_{5}$ formally self-dual \\ \cline{2-3}
$QR_{5}\left( 31\right) $ & $\left( 31,25^{16},16\right) $ & $[62,32,16]_{5}$
\\ \cline{2-3}
$\overline{QR_{5}\left( 31\right) }$ & $\left( 32,25^{16},18\right) $ & $%
[64,32,18]_{5}^{\ast }$ optimal self-dual \\ \cline{2-3}
$QR_{7}\left( 3\right) $ & $\left( 3,49^{2},3\right) $ & $[6,4,3]_{7}^{\ast
} $ \\ \cline{2-3}
$\overline{QR_{7}\left( 3\right) }$ & $\left( 4,49^{2},4\right) $ & $%
[8,4,4]_{7}$ self-dual \\ \cline{2-3}
$QR_{7}\left( 19\right) $ & $\left( 19,49^{10},11\right) $ & $[38,20,11]_{7}$
\\ \cline{2-3}
$\overline{QR_{7}\left( 19\right) }$ & $\left( 20,49^{10},13\right) $ & $%
[40,20,13]_{7}$ self-dual \\ \hline
\end{tabular}
\end{center}

In the table, $\ast $ denotes that the code has the best possible minimum
distance by \cite{Grassl}. Optimal self-dual codes in the above table
coincides with the quadratic double circulant (QDC) codes constructed by
Gaborit in \cite{Gaborit}. Some examples of extended QR codes for larger
primes are given in the following table;

\begin{center}
Table 2: QR codes for $p=11,13,17,19,23$ and $29$

\begin{tabular}{|l|l|l|}
\hline
& The code over $\mathbb{F}_{p}+v\mathbb{F}_{p}$ & Gray image; over $\mathbb{%
F}_{p}$ \\ \cline{2-3}
$\overline{QR_{11}\left( 5\right) }$ & $\left( 6,121^{3},5\right) $ & $%
[12,6,5]_{11}$ formally self-dual \\ \cline{2-3}
$\overline{QR_{11}\left( 7\right) }$ & $\left( 8,121^{4},7\right) $ & $%
[16,8,7]_{11}$ self-dual \\ \cline{2-3}
$\overline{QR_{11}\left( 19\right) }$ & $\left( 20,121^{10},13\right) $ & $%
[40,20,13]_{11}$ self-dual \\ \cline{2-3}
$\overline{QR_{13}\left( 3\right) }$ & $\left( 8,169^{2},4\right) $ & $%
[8,4,4]_{13}$ self-dual, $C_{13,8,10}$ in \cite{betsumiya2} \\ \cline{2-3}
$\overline{QR_{13}\left( 17\right) }$ & $\left( 18,169^{9},12\right) $ & $%
[36,18,12]_{13}$ formally self-dual \\ \cline{2-3}
$\overline{QR_{17}\left( 13\right) }$ & $\left( 14,289^{14},10\right) $ & $%
[28,14,10]_{17}$ formally self-dual \\ \cline{2-3}
$\overline{QR_{17}\left( 19\right) }$ & $\left( 20,289^{10},13\right) $ & $%
[40,20,13]_{17}$ self-dual \\ \cline{2-3}
$\overline{QR_{19}\left( 3\right) }$ & $\left( 4,361^{2},4\right) $ & $%
[8,4,4]_{19}$ self-dual \\ \cline{2-3}
$\overline{QR_{19}\left( 5\right) }$ & $\left( 6,361^{3},6\right) $ & $%
[12,6,6]_{19}$ formally self-dual \\ \cline{2-3}
$\overline{QR_{23}\left( 7\right) }$ & $\left( 8,529^{4},7\right) $ & $%
[16,8,7]_{23}$ self-dual \\ \cline{2-3}
$\overline{QR_{29}\left( 5\right) }$ & $\left( 6,841^{3},6\right) $ & $%
[12,6,6]_{29}$ formally self-dual \\ \cline{2-3}
$\overline{QR_{29}\left( 7\right) }$ & $\left( 8,841^{4},7\right) $ & $%
[16,8,7]_{29}$ self-dual \\ \hline
\end{tabular}
\end{center}

Self-dual codes up to length $20$ in the above table are not optimal by \cite%
{betsumiya2}.

\end{document}